\documentclass[letterpaper, 10 pt]{ieeeconf}
\usepackage{cite}
\usepackage{amsmath,amssymb,amsfonts}
\usepackage{algorithmic}
\usepackage{graphicx}
\usepackage{textcomp}

\usepackage{tikz}
\usetikzlibrary{shapes,arrows}
\usepackage{nicefrac}
\usepackage{mathtools, cuted}
\usepackage{lipsum}
\usepackage{float}
\usepackage{afterpage}
\usepackage{placeins}
\usepackage{tabularx}
\usepackage{balance}
\usepackage{textcomp}%
\usepackage{booktabs}
\usepackage{multirow}
\usepackage{accents}
\usetikzlibrary{patterns,shapes,arrows}
\usetikzlibrary{positioning}
\allowdisplaybreaks

\newcommand\copyrighttext{%
  \footnotesize \textcopyright 2019 IEEE. Personal use of this material is permitted. Permission from IEEE must be obtained for all other uses, in any current or future media, including reprinting/republishing this material for advertising or promotional purposes, creating new collective works, for resale or redistribution to servers or lists, or reuse of any copyrighted component of this work in other works.}
\newcommand\copyrightnotice{%
\begin{tikzpicture}[remember picture,overlay]
\node[anchor=south,yshift=10pt] at (current page.south) {\fbox{\parbox{\dimexpr\textwidth-\fboxsep-\fboxrule\relax}{\copyrighttext}}};
\end{tikzpicture}%
}

\tikzstyle{decision} = [diamond, draw, fill=blue!20,
    text width=7em, text badly centered, inner sep=0pt]
\tikzstyle{decision1} = [diamond, draw, fill=blue!20,
    text width=7em, text badly centered, inner sep=0pt, node distance = 2.5cm]
\tikzstyle{block} = [rectangle, draw, fill=blue!20,
    text width=13em, text centered, rounded corners, minimum height=3em]
\tikzstyle{block2} = [rectangle, draw, fill=blue!20,
    text width=5em, text centered, rounded corners, minimum height=3em, node distance = 4cm]
\tikzstyle{block3} = [rectangle, draw, fill=red!20,
    text width=5em, text centered, rounded corners, minimum height=3em, node distance = 4cm]
\tikzstyle{line} = [draw, very thick, color=black!50, -latex']
\tikzstyle{cloud} = [draw, ellipse,fill=red!20, node distance=2.5cm, minimum height=2em]

\newtheorem{definition}{Definition}
\newtheorem{proposition}{Proposition}

\begin{document}

\title{External Constraint Handling for Solving Optimal Control Problems with Simultaneous Approaches and Interior Point Methods}

\author{Yuanbo Nie and Eric C. Kerrigan \IEEEmembership{Senior Member, IEEE}
\thanks{Yuanbo Nie and Eric C. Kerrigan are with the Department of Aeronautics, Imperial College London, SW7~2AZ, U.K. {\tt\small yn15@ic.ac.uk}, {\tt\small 
e.kerrigan@imperial.ac.uk}}%
\thanks{Eric C. Kerrigan is also with the Department of Electrical \& Electronic Engineering, Imperial College London, London SW7~2AZ, U.K.}%
\thanks{Accepted version to be published in: IEEE Control Systems Letters}%
}

\maketitle
\copyrightnotice

\begin{abstract}
Inactive constraints do not contribute to the solution of an optimal control problem, but increase the problem size and burden the numerical computations. We present a novel strategy for handling inactive constraints efficiently by systematically removing the inactive and redundant constraints.  The method is designed to be used together with  simultaneous approaches under a mesh refinement framework, with mild assumptions that the original problem has feasible solutions, and the initial solve of the problem is successful.  The method is tailored for interior point-based solvers, which are known to be very sensitive to the choice of initial points in terms of feasibility. In the example problem shown, the proposed scheme achieves more than a 40\% reduction in computation time. 
\end{abstract}


\section{Introduction} \label{sec:Intro}
Optimal control has been very popular for a wide range of applications, thanks to its  ability  to handle various types of constraints systematically. When formulating the optimal control problem (OCP), it is common practice to impose a large number of constraints to ensure all mission specifications are fulfilled. However, for the solution obtained, it is often the case that only a small subset of the imposed  inequality  constraints will actually be active. Furthermore, even for the ones in this small subset, the duration for which each constraint is active is generally much shorter than the time dimension of the OCP.  One example would be for the design of flight control systems: although all limits of the flight envelope need to be specified in the problem formulation for safety requirements, only in rare (abnormal) situations is it the case that some limits will be reached.   

In numerical optimal control, the OCPs are transcribed into sparse nonlinear programming (NLP) problems.  A distinction can be made here between simultaneous and sequential approaches, depending on whether all or just the control trajectories are discretized as decision variables \cite{biegler2007overview}. For this work, we focus on the simultaneous approach. 

The main computational overheads for solving the NLP problems are directly related to the number of decision variables and constraints. Thus there exist significant computational benefits to exclude inactive constraints in the problem formulation.  One possibility is to only include the constraints that are determined to be active. Based on this, an external strategy for the handling of path constraints with active-set based NLP solvers has been proposed in~\cite{chung2009external}. The idea is to first solve the unconstrained problem and determine which constraints are likely to be active based on constraint violations. These constraints are then added in the OCP and the problem is repetitively solved until all original constraints are satisfied. However, a fundamental problem arises when implementing the same idea on interior point method (IPM) based solvers, since good performance hinges on the initial point to be feasible, or at least close to feasible \cite{gondzio2016crash}. 

Another option is to remove constraints that are inactive. Removal of constraints for model predictive control (MPC) has been studied to accelerate computations for linear MPC~\cite{jost2016constraint}, tube-based robust linear MPC  \cite{jost2015accelerating} and recently nonlinear MPC \cite{constraintRemovalNMPC}, with computational benefits clearly demonstrated. However, all of these are based on a quadratic regulation cost, making their application specifically aimed at receding horizon control of regulation tasks.



In this paper, we introduce an external constraint handling (ECH) strategy that is tailored to IPM-based NLP solvers for solving a  variety of OCPs. Implemented together with mesh refinement (MR) schemes, constraints that do not contribute to the solution are systematically removed in the problem formulation. Special attention is paid to ensure feasibility of the initial point. As a result, significant computational savings can  be achieved. Section~\ref{sec:OptimizationBasedControl} gives an introduction to numerical optimal control with direct collocation. This is followed by a discussion of the proposed ECH strategy in Section~\ref{sec:ExternalConstraintHandling}. A flight control example is presented in Section~\ref{sec:ExampleFlightProfile} to demonstrate the computational benefits. 

\section{Numerical Optimal Control}
\label{sec:OptimizationBasedControl}
Generally speaking, optimization-based control requires the solution of  OCPs expressed in the general Bolza form:
\begin{subequations}
\label{eqn:OCPBolza}
\begin{equation}
\min_{x,u,p,t_0,t_f} \Phi(x(t_0),t_0,x(t_f),t_f,p)
+\int_{t_0}^{t_f} L(x(t),u(t),t,p) dt
\end{equation}
subject to
\begin{align}
\dot{x}(t)=f(x(t),u(t),t,p),\ &\forall t \in [t_0,t_f] \label{eqn:OCPBolzaDynamics}\\
c(x(t),u(t),t,p)\le 0,\ &\forall t \in [t_0,t_f] \label{eqn:OCPBolzaPathConstraint}\\
\phi(x(t_0),t_0,x(t_f),t_f,p) =0,\ &
\end{align}
\end{subequations}
 with $x: \mathbb{R} \rightarrow \mathbb{R}^n$ is the state trajectory of the system, $u: \mathbb{R} \rightarrow \mathbb{R}^m$ is the control input trajectory,   $p \in \mathbb{R}^s$ are static parameters, $t_0 \in \mathbb{R}$ and $t_f \in \mathbb{R}$ are the initial and terminal time.  $\Phi$ is the Mayer cost functional ($\Phi$: $\mathbb{R}^n \times \mathbb{R} \times \mathbb{R}^n \times \mathbb{R} \times \mathbb{R}^s \to \mathbb{R}$), $L$ is the Lagrange cost functional ($L:\mathbb{R}^n \times \mathbb{R}^m \times \mathbb{R} \times \mathbb{R}^s \to \mathbb{R}$), $f$ is the dynamic constraint ($f:\mathbb{R}^n \times \mathbb{R}^m \times \mathbb{R} \times \mathbb{R}^s \to \mathbb{R}^n$), $c$ is the path constraint ($c:\mathbb{R}^n \times \mathbb{R}^m \times \mathbb{R} \times \mathbb{R}^s \to \mathbb{R}^{n_g}$) and $\phi$ is the boundary condition ($\phi:\mathbb{R}^n \times \mathbb{R} \times \mathbb{R}^n \times \mathbb{R} \times \mathbb{R}^s \to \mathbb{R}^{n_q}$).

In practice, most optimal control problems formulated as~\eqref{eqn:OCPBolza} need to be solved with numerical schemes.  Compared to sequential methods,  simultaneous methods have some  advantages with regards to computational efficiency, as well as in the treatment of path constraints and unstable dynamics~\cite{biegler2007overview}. In this paper, we will demonstrate our proposed ECH strategy with the simultaneous approach of direct collocation. 

\subsection{Direct collocation methods}
\label{sec:DirectTranscriptionMethod}
Direct collocation methods can be categorized into fixed-order $h$ methods \cite{betts2010practical}, and variable-order $p$/$hp$ methods \cite{fahroo2008advances,liu2014hp}. Here, we only provide a high level overview. For a mesh of size $N:=\sum_{k=1}^K N^{(k)}$, the states can be approximated as

\begin{equation*}
\label{eqn:LGRStateApproximation}
x^{(k)}(\tau) \approx \bar{x}^{(k)}(\tau) := \sum_{j=1}^{N^{(k)}}\mathcal{X}_j^{(k)}\mathcal{B}_{j}^{(k)}(\tau),
\end{equation*}

within mesh interval $k$ $\in$ $\{1,\ldots, K\}$, where $N^{(k)}$ denotes the number of collocation points  for  interval $k$ and $\mathcal{B}_{j}^{(k)}(\cdot)$ are basis functions. For typical $h$ methods, $\tau \in \mathbb{R}^{N}$ takes values on the interval $[0,1]$ representing $[t_0,t_f]$, and $\mathcal{B}_{j}^{(k)}(\cdot)$ are elementary B-splines of various orders. For $p$/$hp$ methods, $\tau$ $\in$ $[-1,1]$ and $\mathcal{B}_{j}^{(k)}(\cdot)$ are Lagrange interpolating polynomials.  We use $X_j^{(k)}$ and $U_j^{(k)}$ to represent the approximated states and inputs at collocation points, e.g.\ $X_j^{(k)}=\bar{x}^{(k)}(\tau_j^{(k)}) \in \mathbb{R}^{n}$, where $\tau_j^{(k)}$ is the $j^\text{th}$ collocation point in mesh interval~$k$.

Consequently, the OCP~\eqref{eqn:OCPBolza} can be approximated by
\begin{subequations}
\label{eqn:LGRStateApproximationCost}
\begin{multline}
\min_{X,U,p,t_0,t_f}  \Phi(X_1^{(1)},t_0,X_{f}^{(K)},t_f,p)\\
+\sum_{k=1}^{K}\sum_{i=1}^{N^{(k)}} w_i^{(k)} L(X_i^{(k)},U_i^{(k)},\tau_i^{(k)},t_0,t_f,p)
\label{eqn:LGRStateApproximationCostDefect}
\end{multline}
for $i=1,\ldots,N^{(k)}$ and $k=1,\ldots,K$, subject to, 
\begin{align}
\sum_{j=1}^{N^{(k)}}\mathcal{A}_{ij}^{(k)}X_j^{(k)}+\mathcal{D}_{i}^{(k)}f(X_i^{(k)},U_i^{(k)},\tau_i^{(k)},t_0,t_f,p) = & 0 \\
\label{eqn:LGRStateApproximationPathConstraint}
c(X_i^{(k)},U_i^{(k)},\tau_i^{(k)},t_0,t_f,p)\le & 0  \\
\phi(X_1^{(1)},t_0,X_{f}^{(K)},t_f,p) =& 0
\end{align}
\end{subequations}
where $w_j^{(k)}$ are the quadrature weights for the chosen discretization, $\mathcal{A}^{(k)}$ is the numerical differentiation matrix with element $(i,j)$ denoted by $\mathcal{A}_{ij}^{(k)}$  and $\mathcal{D}_i^{(k)}$ is a row vector.   

The discretized problem can then be solved with off-the-shelf NLP solvers. The NLP solver generates a discretized solution $Z \coloneqq (X, U, p, \tau, t_0, t_f)$ as sampled data points. Interpolating splines may be used to construct an approximation of the continuous-time optimal trajectory $t\mapsto \tilde{z}(t) \coloneqq (\tilde{x}(t), \tilde{u}(t), t, p)$. The quality of the interpolated solution needs to be assured through error analysis, assessing the level of accuracy and constraint satisfaction  at a much higher resolution than the discretization mesh.  

If necessary, appropriate modifications must be made to the discretization mesh, until the solutions obtained with the new mesh fulfills all predefined error tolerance levels (e.g. the absolute local error $\eta_{tol}$ and the absolute local constraint violation $\varepsilon_{c_{tol}}$).  This process of MR is crucial in solving large-scale problems efficiently. For instance, it took 6 MR iterations for the example problem to be solved to a specific tolerance level. This level of accuracy was not achievable with any uniform mesh using the same desktop computer.  

\section{External Constraint Handling 
} \label{sec:ExternalConstraintHandling}

\subsection{Active and inactive constraints}
A constraint \eqref{eqn:OCPBolzaPathConstraint} is considered \emph{active} if its presence influences the solution $z^{\ast}(\cdot) \coloneqq (x^{\ast}(\cdot),u^{\ast}(\cdot),p^{\ast},t_0^{\ast},t_f^{\ast})$. A constraint is \emph{inactive} if it can be removed without affecting the solution. To clarify, consider a simplified problem:
\begin{equation*}
\label{eqn:SimplifiedOptProblem}
y^*\in \operatorname{arg}\min_{y} \Phi(y) \quad \text{subject to} \quad c(y)\le 0,
\end{equation*}
where we need to identify conditions such that constraints can be determined to be active. The most obvious criteria is when the solution $y^{\ast}$ is at the boundary of $c_i(y)\le 0$, i.e.\ $c_i(y^{\ast}) = 0$. Additionally,  consider the Lagrangian  $\mathcal{L}:=\Phi(y)+\lambda^T c(y)$ and the necessary optimality conditions (Karush-Kuhn-Tucker (KKT) conditions):
\begin{align*}
\begin{split}
\label{eqn:KKTSystem}
\frac{\partial \Phi(y)}{\partial y}\rvert_{y=y^{\ast}}+\lambda\frac{\partial c(y)}{\partial y}\rvert_{y=y^{\ast}}&=0,\\
c(y^{\ast}) \le 0, \quad \quad \lambda \ge 0, \quad \quad \lambda  \circ c & (y^{\ast}) =0,
\end{split}
\end{align*}
with $\circ$ the Hadamard product. From the included complementary slackness condition, we know that for strictly positive Lagrange multipliers ($\lambda_i > 0$), the corresponding solution will have $c_i(y^{\ast}) = 0$, i.e.\ the constraint is active. 


\subsection{Identifying active constraints in  optimal control}
\label{subsec:identifyActiveNumerialOCP}

 Theoretically, the above-mentioned analysis applies only to the continuous OCP formulation \eqref{eqn:OCPBolza}. Additional challenges will arise in practice when solving the discretized problem \eqref{eqn:LGRStateApproximationCost} numerically:  the NLP solver will only return the values of the discretized state $X$, input $U$, and Lagrange multipliers $\Lambda$ at collocation points. 

To estimate the constraint activation status in-between collocation points, a criteria can be introduced based on the interpolated continuous trajectory $\tilde{z}$.  By definition, inequality constraints are active if the magnitude of the differences between the actual constraint $c_l(\tilde{z}(\cdot))$ and the user-defined constraint bounds are zero. Due to numerical  inaccuracies, however, there will always be a remainder. Thus, we consider a constraint to be potentially active if this difference is smaller than the constraint violation tolerance $\epsilon_{c_{tol}}$.

Note that the word \emph{potentially} is used to emphasise that, for numerical schemes under limited machine precision, no concrete determination of constraint active status can be made. On the other hand, we know that only if the identified inactive constraints are truly inactive, then can they be removed from the OCP without affecting the solutions. Thus, it would be much more preferable to erroneously identify inactive constraints as active, than the opposite situation.

For this reason, we also use the multiplier information to enforce a larger (more conservative) selection of potentially active constraints. Here, a similar numerical challenge arises: with limited machine precision, even when the corresponding constraints are inactive, the multiplier values are rarely truly equal to zero.
  To identify the regions where the constraints are likely to be active, the numerical multiplier data $\Lambda$ is first normalized between 0 and 1 for each constraint $c_l(Z) \le 0$. Signal processing algorithms can be used to identify different intervals where the behaviour of Lagrange multipliers have significant changes,  for example using the \textsc{Matlab} \texttt{findchangepts} function.  


For each identified interval $\mathcal{T}_i$, the mean value of the normalized multipliers ($\bar{\Lambda}_{\mathcal{T}_i}$) is calculated and compared based on the following criteria:
\begin{equation*}
\begin{array}{ll}
\text{if } \bar{\Lambda}_{\mathcal{T}_i} \ge \zeta & \text{constraint potentially active in interval }\mathcal{T}_i\\
\text{otherwise} & \text{constraint potentially inactive in interval }\mathcal{T}_i\\
\end{array}
\end{equation*} 
with $\zeta$ a threshold parameter. 

To sum up, the following definition is used to determine whether the constraints are potentially active or potentially inactive at different collocation points.
\begin{definition}
\label{def: potentiallyactive}
A constraint  $c_l(X_i,U_i,\tau_i,t_0,t_f,p) \le 0$ is  \emph{potentially active} at time  $t_i:=t_i(\tau_i,t_0,t_f)$ if one of the following criteria is met:
\begin{itemize}
\item Between adjacent collocation points ($t \in [t_{i-1},t_{i+1}]$), $c_l(\tilde{x}(t),\tilde{u}(t),t,p) \ge  - \epsilon_{c_{tol}}$ holds, with $\epsilon_{c_{tol}} > 0$.
\item $\bar{\Lambda}_{\mathcal{T}_i} \ge \zeta$, with $t_i \in \mathcal{T}_i$.
\end{itemize}
Otherwise, a constraint  is  \emph{potentially inactive} at time~$t_i$.
\end{definition}

In addition to identifying the time instances at which certain constraints may be potentially inactive, it is also preferable to determine the sets of constraints that never become active at all times. 

\begin{definition}
\label{def: potentiallyredundant}
A constraint $c_l(X_i,U_i,\tau_i,t_0,t_f,p) \le 0$ is  \emph{potentially redundant} if for all $t_i:=t_i(\tau_i,t_0,t_f)$ with $i=1,\ldots, N$, the constraints $c_l(X_i,U_i,\tau_i,t_0,t_f,p) \le 0$ are potentially inactive. Otherwise, this set of constraints is \emph{potentially enforced}. 
\end{definition}

\subsection{Initialization for interior point methods}
\label{subsec:initialIPM}
Interior point methods (IPMs) for solving NLPs were introduced in the early 1960s
\cite{fiacco1963programming,fiacco1964sequential,fiacco1966extensions} and have became very popular in numerical optimal control. The  idea is to augment the objective function with barrier functions of constraints in order to enforce their satisfaction. Potential solutions will iterate only in the feasible region following the so-called central path, resulting in a very efficient algorithm. 

Standard interior point methods are  sensitive to the choice of a starting point. To ensure that the initial guess is strictly feasible with respect to constraints, various initialization methods have been developed (e.g.\ \cite{mehrotra1992implementation} and the collective study in \cite{betts2010practical}) and implemented in modern solvers.

To ensure reliable and efficient computation of the initialization algorithm, as well as the subsequent NLP iterations, several criteria~\cite{gondzio2016crash} can be formulated regarding ideal initial points for IPMs. The ideal initial point should:
\begin{itemize}
\item  satisfy or be close to primal and dual feasibility,
\item  be close to the central path,
\item  be as close to optimality as possible.
\end{itemize}


Because of these characteristics, external constraint handling schemes developed for active-set based SQP solvers, such as \cite{chung2009external}, are not  suitable for IPM-based solvers. By first solving the unconstrained problem and gradually adding constraints based on the constraint violation error, the solution of previous solves will all be infeasible for the new OCP formulation, and the solution may undergo drastic changes as well. For IPM-based NLP solvers, this would lead to a higher computational overhead for initialization, as well as higher chances for the iterations to frequently enter the slow and unreliable feasibility restoration phase. 

\subsection{Proposed Scheme for Constraint Handling}
Based on the criteria presented in Section \ref{subsec:identifyActiveNumerialOCP} and the characteristics of IPMs as discussed in Section \ref{subsec:initialIPM}, a strategy for efficiently handling constraints in OCPs solved with IPM-based NLP solvers is proposed, with the work-flow presented in Figure~\ref{fig:ECHFlowChart}. The approach is called \emph{external}, since the modifications to the OCP are made at the MR iteration level, instead of during the NLP iterations.

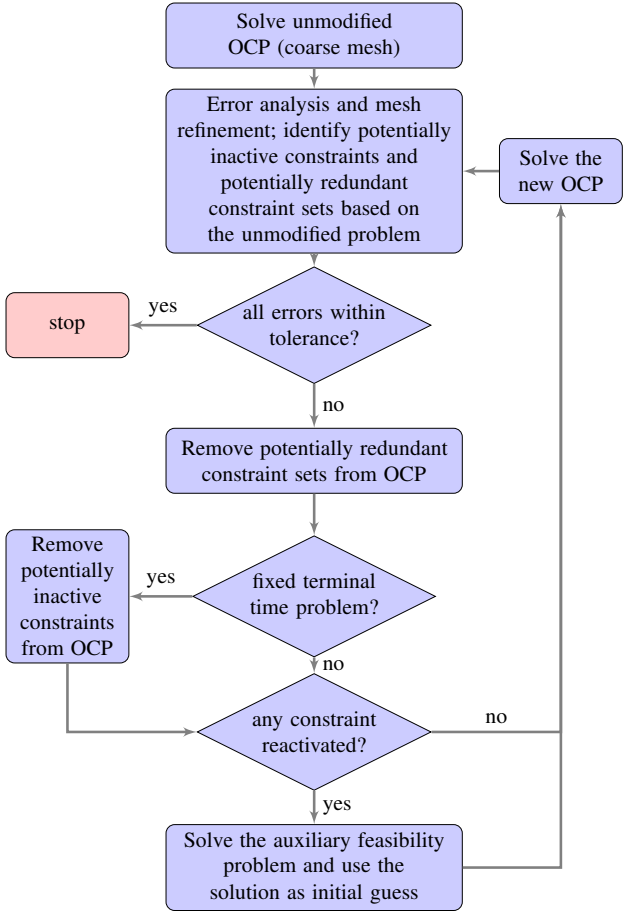
\begin{figure}
\centering
\scalebox{0.82}{\begin{tikzpicture}[align=center, scale=1, node distance = 2.2cm and 1cm]
    \node [block] (initSolve) {Solve unmodified OCP (coarse mesh)};
    \node [block, below of=initSolve] (postCheck) {Error analysis and mesh refinement; identify potentially inactive constraints and potentially redundant constraint sets based on the unmodified problem};
    \node [decision1, below of=postCheck, aspect=2] (errorCheck) {all errors within tolerance?};
    \node [block, below of=errorCheck] (removeRed) {Remove potentially redundant constraint sets from OCP};
    \node [decision, below of=removeRed, aspect=2, scale=1] (decideFixtf) {fixed terminal time problem?};
    \node [block2, left of=decideFixtf] (removeActive) {Remove potentially inactive constraints from OCP};
    \node [decision, below of=decideFixtf, aspect=2] (checkReActive) {any constraint reactivated?};
    \node [block, below of=checkReActive] (solveAuxOCP) {Solve the auxiliary feasibility problem and use the solution as initial guess};
    \node [block2, right of=postCheck] (resolve) {Solve the new OCP};
    \node [block3, left of=errorCheck] (stop) {stop};
    
    \path [line] (initSolve) -- (postCheck);
    \path [line] (postCheck) -- (errorCheck);
	\path [line] (errorCheck) -- node [color=black,above] {yes}(stop);
    \path [line] (errorCheck) -- node [color=black,right] {no}(removeRed);
    \path [line] (removeRed) -- (decideFixtf);
    \path [line] (decideFixtf) -- node [color=black,above] {yes} (removeActive);
    \path [line] (decideFixtf) -- node [color=black,right] {no} (checkReActive);
    \path [line] (removeActive) |- (checkReActive);
    \path [line] (checkReActive) -- node [color=black,right] {yes} (solveAuxOCP);
    \path [line] (checkReActive) -| node [color=black,near start, above] {no} (resolve);
    \path [line] (solveAuxOCP) -| (resolve);
    \path [line] (resolve) -- (postCheck);
\end{tikzpicture}}
\caption{Overview of the proposed external constraint handling scheme}
\label{fig:ECHFlowChart}
\end{figure}

The unmodified OCP is first solved on the initial coarse mesh. Even with all constraint equations included, the computation time will still be quite low at this stage due to the small problem size. Once the solution is obtained, potentially inactive constraints and potentially redundant constraint sets can be identified,  based on Definitions~\ref{def: potentiallyactive} and \ref{def: potentiallyredundant}, with potentially redundant constraints directly excluded from the OCP formulation. Furthermore, if the problem has a fixed terminal time, i.e.\ the time instance corresponding to a mesh point will not change, then potentially inactive constraints in the potentially enforced constraint sets may also be removed.

Recall that it is preferable to erroneously identify an inactive constraint as potentially active, rather than the opposite. It is therefore often a good idea in practice to enlarge the intervals with potential constraint activation by an interval of length $\beta$ in each direction, with $\beta$ either fixed or adapting during the MR process. This adaptation also guarantees the convergence of the overall scheme, i.e.\ in the worst case, $\beta$ can be sufficiently large to impose the constraints for the whole trajectory, with the original problem recovered.

In-between MR iterations, special attention must be made to constraints and constraint sets that were determined to be potentially inactive or redundant in the previous solves. If they never become active or enforced again, the constraint removal process may  continue  until MR is converged. But there will be the chance, after refining the mesh,  the constraint violation error analysis dictates  that certain constraints and constraint sets that have been removed earlier may become potentially active or enforced again.  If this happens, they need to be included again to ensure that the solution of the modified OCP is equivalent to the unmodified problem. 

Note that the previous solution will no longer be a feasible initial guess for the new problem formulation. To assist the subsequent solve of NLPs, the following auxiliary feasibility problem (AFP) can be solved before proceeding:
\begin{subequations}
\label{eqn:ECHFeasibilityProblem}
\begin{equation}
J^*:=\min_{X,U,p,t_0,t_f} \sum_{l=1}^{n_g} s_l
\end{equation}
subject to, for $i=1,\ldots,N^{(k)}$ and $k=1,\ldots,K$,
\begin{align}
\sum_{j=1}^{\tilde{N}}\mathcal{A}_{ij}^{(k)}X_j^{(k)}+\mathcal{D}_{i}^{(k)}f(X_i^{(k)},U_i^{(k)},\tau_i^{(k)},t_0,t_f,p) = & 0 \\
\label{eqn:ECHFeasibilityPathConstraint}
c(X_i^{(k)},U_i^{(k)},\tau_i^{(k)},t_0,t_f,p)\le s,   \text{ with }  s \ge & 0\\
\phi(X_1^{(1)},t_0,X_{K}^{(K)},t_f,p) =& 0
\end{align}
\end{subequations}
with $s \in \mathbb{R}^{n_g}$ slack variables. The initial guess for the AFP will be $\tilde{Z} \coloneqq (\tilde{X}, \tilde{U}, p, \tilde{\tau}, t_0, t_f)$, the values of the interpolated solution $\tilde{z}$ at the collocation points of the refined mesh. 

\subsection{Properties of the external constraint handling strategy}

It is possible to derive proofs of feasibility and optimality invariance for the removal of constraints on a given discretization mesh. However, with the size of the mesh changing throughout the refinement process, the analysis of errors, and identification of constraint activation status (Section \ref{subsec:identifyActiveNumerialOCP}) are all subject to considerable uncertainties. When a constraint or constraint set must be included again in the problem, it will be challenging to ensure that the subsequent OCP solve can be supplied with a feasible initial guess. The introduction of the AFP is the answer to this challenge, with its solution guaranteed to be a feasible point for the corresponding original OCP.

\begin{proposition}
\label{prop: AOCPSolutionFeasibility}
If the original OCP \eqref{eqn:LGRStateApproximationCost} has feasible points, then a  solution to the auxiliary feasibility problem~\eqref{eqn:ECHFeasibilityProblem} will be a feasible point of \eqref{eqn:LGRStateApproximationCost} on the same discretization mesh, and~\eqref{eqn:ECHFeasibilityProblem} will have corresponding objective value $J^{\ast}=0$.
\end{proposition}
\begin{proof}
If $Z \coloneqq (X,U,\tau,t_0,t_f,p)$ is a feasible point of \eqref{eqn:LGRStateApproximationCost}, then \eqref{eqn:LGRStateApproximationPathConstraint} must hold. With $s \ge 0$, the solution for the AFP~\eqref{eqn:ECHFeasibilityProblem} will be the situation where $\sum s_l = 0$, and \eqref{eqn:LGRStateApproximationPathConstraint} guarantees the existence of such a solution.
\end{proof}

Now, for the very same reason, we need to obtain a suitable initial guess for the slack variables $s$ in the AFP. One possible way is by calculating the constraint violation errors of the interpolated solutions on the refined mesh. 

\begin{proposition}
\label{prop: AOCPFeasibility}
Define $\tilde{s} \in \mathbb{R}^{N \times n_g}$ as the absolute local constraint violation error $\epsilon_c(t)$ calculated at the collocation points of the refined mesh, with the updated initial guess $\tilde{Z} \coloneqq (\tilde{X}, \tilde{U}, p, \tilde{\tau}, t_0, t_f)$. For any set $\{\bar{s} \in \mathbb{R}^{n_g} \mid \bar{s_l} \ge \max_{i=1,\dots,N}(\tilde{s}_{i,l}),l=1,\ldots,n_g\}$ implemented as the initial guess for $s$, the AFP \eqref{eqn:ECHFeasibilityProblem} will have a strictly feasible initial point with respect to the constraints \eqref{eqn:ECHFeasibilityPathConstraint}.
\end{proposition}
\begin{proof}
$\tilde{s}_{i,l}:=|\min(-c_l(\tilde{X}_i$, $\tilde{U}_i$, $p$, $\tilde{\tau}_i$, $t_0$, $t_f),0)|$ by definition, for all $i=1,\ldots,N$ and $l=1,\ldots,n_g$,
\begin{itemize}
\item if $c_l(\tilde{X}_i$, $\tilde{U}_i$, $p$, $\tilde{\tau}_i$, $t_0$, $t_f) < 0$, i.e.\ the constraint is satisfied and the solution is not on the boundary, then $\tilde{s}_{i,l}=0$, thus $c_l(X_i,U_i,\tau_i,t_0,t_f,p) < \tilde{s}_{i,l}$ holds.
\item if $c_l(\tilde{X}_i$, $\tilde{U}_i$, $p$, $\tilde{\tau}_i$, $t_0$, $t_f) = 0$ (constraint satisfied and solution is on the boundary) or $c_l(\tilde{X}_i$, $\tilde{U}_i$, $p$, $\tilde{\tau}_i$, $t_0$, $t_f) > 0$ (constraint violation occurs), then $c_l(X_i,U_i,\tau_i,t_0,t_f,p) = \tilde{s}_{i,l}$ holds.
\end{itemize} 
Therefore, $c_l(X_i,U_i,\tau_i,t_0,t_f,p) \le \tilde{s}_{i,l}$ will always be true. From $\bar{s}_l \ge \max_{i=1,\dots,N}(\tilde{s}_{i,l})$, it can  be concluded that $c_l(X_i,U_i,\tau_i,t_0,t_f,p) \le \bar{s}_l$ holds.
\end{proof}

We can now show that, except for the initial solve, all subsequent solves will have feasible initial guesses.

\begin{proposition}
\label{prop: OCPStrictlyFeasible}
If the unmodified OCP has feasible points, and the initial solve of the discretized OCP has been successful, then all subsequent solves of OCPs and AFPs with mesh refinement schemes and the proposed external constraint handling method will have a feasible initial point with respect to the constraints \eqref{eqn:LGRStateApproximationPathConstraint} or \eqref{eqn:ECHFeasibilityPathConstraint}.
\end{proposition}
\begin{proof}
For any interpolated solution $\tilde{Z} \coloneqq (\tilde{X}, \tilde{U}, p, \tilde{\tau}, t_0, t_f)$ on the new mesh, if \eqref{eqn:LGRStateApproximationPathConstraint} is not satisfied, then the corresponding AFP will be solved and Proposition~\ref{prop: AOCPFeasibility} ensures that AFP will have a feasible initial guess. From Proposition~\ref{prop: AOCPSolutionFeasibility} the solution of the AFP will be a feasible initial guess for the subsequent OCP solve.
\end{proof}

\subsection{A practically more efficient alternative implementation}
\label{subsec:alternativeImplementation}
The proposed external constraint handling scheme can guarantee feasible initial points under conditions stated in Proposition~\ref{prop: OCPStrictlyFeasible}. Nevertheless, it is not efficient in practice. The frequent solve of AFPs are not only time consuming, but are often not necessary.

Recall the conditions regarding ideal initial guesses for IPM methods. It is not necessary to satisfy primal and dual feasibility --- rather, one only needs to be close to fulfillment. In practice, the computational performance of a modern IPM that uses near-feasible initial guesses is very much comparable to using feasible initial points. In addition, constraint satisfaction for  simple bounds can be computationally much easier to achieve by the NLP solver, thus there is no need to enforce those through the solve of an AFP. 

 Thus, a practically more efficient version of the external handling scheme can be formulated, by restricting the conditions for solving the AFP to the MR iteration when a potentially redundant path constraint set turns into a potentially enforced path constraint set.


\section{Example}
\label{sec:ExampleFlightProfile}

To demonstrate the computational benefits of the proposed ECH scheme, we show an problem that is relatively large in the horizon length. The task involves finding a fuel-optimal flight path of a commercial aircraft where authorities have identified five non-flight zones (NFZ) for the aircraft to avoid.


From simple flight mechanics with a flat earth assumption, both the longitudinal and lateral motion of the aircraft can be described by the  dynamic equations
\begin{subequations}
\label{eqn:DynamicsF50Flight}
\begin{align*}
\dot{h}(t) =&v_{T}(t)\sin(\gamma(t)) \\
\dot{POS_N}(t) =&v_{T}(t)\cos(\gamma(t))\cos(\chi(t)) \\
\dot{POS_E}(t) =&v_{T}(t)\cos(\gamma(t))\sin(\chi(t)) \\
\begin{split}
\dot{v}_{T}(t)=&\frac{1}{m(t)}(T(v_{C}(t),h(t),\Gamma(t))\\
&-D(v_{T}(t),h(t),\alpha(t))-m(t)g\sin(\gamma(t)))
\end{split}\\
\begin{split}
\dot{\gamma}(t) =&\frac{1}{m(t)v_{T}(t)}(L(v_{T}(t),h(t),\alpha(t))\cos(\phi(t))\\
&-m(t)g\cos(\gamma(t)))\end{split}\\
\dot{\chi}(t) = & \frac{L(v_{T}(t),h(t),\alpha(t))\sin(\phi(t))}{\cos(\gamma(t))m(t)v_{T}(t)}\\
\dot{m}(t)=&FF(h(t),v_{C}(t),\Gamma(t))
\end{align*}
\end{subequations}
with $h$ the altitude [m], $POS_{N}$ and $POS_{E}$ the north and east position [m], $v_{T}$ the true airspeed [m/s],  $\gamma$ the flight path angle [rad], $\chi$ the tracking angle  [rad], and $m$ the mass [kg]. $T$, $L$, $D$ are the thrust, lift and drag forces. $FF$ is the fuel flow model, requiring an input of calibrated airspeed $v_{C}$, which can be related to $v_{T}$ via a conversion.



Additionally, $g=9.81$\,m/s$^2$ is  gravitational acceleration. We have three control inputs,  the roll angle $\phi$ in [rad],  the throttle settings $\Gamma$ normalized between 0 and 1, and  the angle of attack $\alpha$ in [rad]. Further details of the modelling of a Fokker 50 aircraft can be obtained from \cite{F50Performance}.

The avoidance of NFZs can be implemented with the following path constraints 
\begin{align*}
\begin{split}
(POS_{N}(t)-POS_{N_{N}})^2+(POS_E(t)-POS_{E_{N}})^2 \ge r_{N}^2
\end{split}
\end{align*}
\par\noindent with $POS_{N_{N}}$ and $POS_{E_{N}}$ the north and east position of the center of the non-flight zones, and $r_{N}$ the radius. 

The problem will have the boundary cost $\Phi(x(t_0),t_0,x(t_f),t_f,p)=-m(t_f)$ (maximize the mass at the end of the flight, with fixed $t_f=7475$\,s), subject to the dynamics and path constraints. Furthermore, variable simple bounds are imposed together with the boundary conditions.



The OCP is transcribed using the optimal control software \texttt{ICLOCS2} \cite{ICLOCS2} with Hermite-Simpson discretization, and solved with IPM-based NLP solver \texttt{IPOPT} \cite{wachter2006implementation} (version 3.12.4).  All computation results shown were obtained on an Intel Core i7-4770 computer with 16G of RAM. Figure~\ref{fig:F50Flight} illustrates the results solved to a user-defined tolerance.

\begin{figure}[tb]
	\begin{center}
	\includegraphics[width=.33\textwidth]{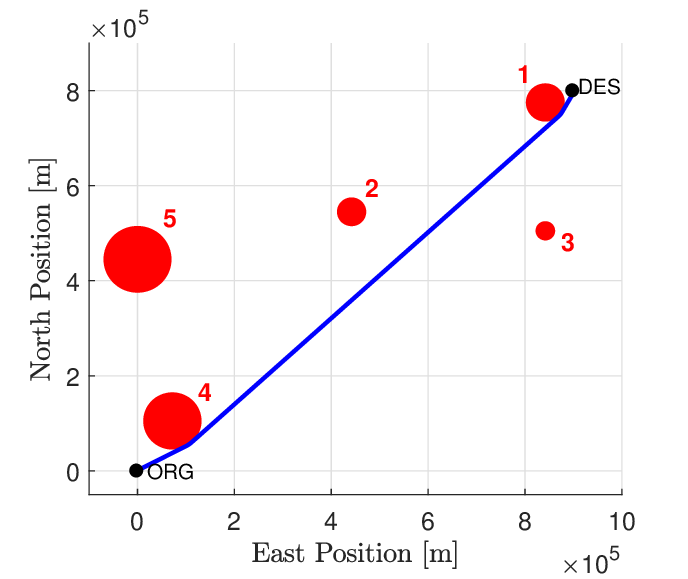}
	\caption{The fuel-optimal flight profile for the example problem}
	\label{fig:F50Flight} 
	\end{center}
\end{figure}



\begin{table*}[hbt]
  \small
	\begin{center}
	\caption{External constraint handling history: Aircraft Flight Profile ($t_0=$0\,s, $t_f=7475$\,s, $\beta=0$\,s)}
	 \label{tab: ExF50FlightECHHistory}
		\begin{tabular}{>{}c|>{}c|>{}c|>{}c|>{}c|>{}c|>{}c}
		 & \multicolumn{6}{>{}c}{Constraint Activation Intervals Implemented in the OCP [$s$]}\\ \cline{2-7}
		 & MR Iteration 1 & MR Iteration 2 & MR Iteration 3 & MR Iteration 4 & MR Iteration 5 & MR Iteration 6\\ 
		 & (K = 40) & (K = 81) & (K = 136) & (K = 176) & (K = 207) & (K = 256)\\ \hline
		 NFZ 1 & [$t_0$, $t_f$] & [6900 7092] & [6996 7114] & [7056 7162] & [7071 7140] & [7074 7150] \\ \hline
		 NFZ 2/3/5 & [$t_0$, $t_f$] & $\emptyset$ & $\emptyset$ & $\emptyset$ & $\emptyset$ & $\emptyset$ \\ \hline
		 NFZ 4 & [$t_0$, $t_f$] & [671 863] & [743 942] & [767 875] & [774 880] & [774 878] \\ \hline
		\end{tabular} 
	\end{center}
\end{table*}

Using the proposed external constraint handling scheme, we first solved the problem with a worst-case buffer interval setting of $\beta=0$\,s. The history for constraint activation intervals implemented in the OCP are demonstrated in Table~\ref{tab: ExF50FlightECHHistory}. It can be seen that in the initial solve (MR iteration 1), all constraint sets are enforced and all constraints are treated as potentially active. Based only on the solution from this coarse grid, the ECH method correctly identified that the constraint sets related to NFZ~2, 3 and 5 are all potentially redundant. It also determined that constraints related to NFZ~1 are only potentially active near the end of the flight, whereas for NFZ~4 they are at the beginning of the mission. 

In later iterations of the MR, these intervals  had only some minor adjustments. It can be seen that without implementing any buffer interval, we do see occasions where active constraints got erroneously identified as inactive for finer meshes. However, the constraint violation error analysis in the MR process correctly identified these situations and made corrections accordingly. 


Table~\ref{tab: ExF50FlightCompPerf} compares the computational performance of the standard solve, as well as solves with external constraint handling   using the alternative ECH implementation (allowing near-feasible initial guesses) described in Section \ref{subsec:alternativeImplementation}, with two different buffer interval settings. With the worst-case setting of $\beta=0$, the total computation time saw a 29\% reduction, while the number of MR iterations remained the same. Choosing a much more conservative buffer interval setting of $\beta=0.1(t_f-t_0)=747.5$\,s further improved this time reduction to 40\%, due to the fact that initial  guesses were feasible for all later MR iterations.

For real-time applications, it is useful to consider the re-computation time for solving the OCP problem again with the final (refined) discretization mesh, using the obtained solutions as initial guesses. For the ECH method with $\beta=747.5$\,s, the time taken was only half compared to the standard solve. Therefore, the benefits of the proposed scheme can be seen for both off-line and online applications. 

\begin{table}[tb]
  \small
	\begin{center}
	\caption{Computational performance comparison}
	 \label{tab: ExF50FlightCompPerf}
		\begin{tabular}{>{}c|>{}c|>{}c|>{}c}
		 & Standard & With ECH & With ECH\\
		 & Solve & ($\beta=0$\,s) & ($\beta=747.5$\,s)\\ \hline
	     \textbf{Total Comp.} & \multirow{2}{*}{130.54} & 92.14 & 78.27  \\
	      \textbf{Time [s]} & & (29\% lower) & (40\% lower)\\ \hline
	     \textbf{MR Iterations} & 6 & 6 & 6 \\
	      \hline
	     \textbf{Re-comp. } & \multirow{2}{*}{21.82} & 13.50 & 10.78 \\
	      \textbf{Time [s]} & & (38\% lower) & (50\% lower) \\ \hline
	      \textbf{Fuel Used [kg]} & 1787.6 & 1787.6 & 1787.6\\
	      \hline
		\end{tabular} 
	\end{center}
\end{table}

\section{Conclusions}
\label{sec:conlclusions}
A strategy has been developed to systematically identify and handle inactive constraints and redundant constraint sets for numerically solving optimal control problems together with mesh refinement schemes. Unlike previous work that would always result in infeasible initial guesses  for intermediate steps,  the proposed scheme is capable of providing guarantees on the feasibility of initial points  in mesh refinement iterations. The method only requires some mild conditions --- the original OCP to have feasible points, and the initial solve of the discretized OCP to be successful, making it particularly suitable for OCP toolboxes that utilize IPM-based NLP solvers in lowering the computational cost. 

 Due to limitations in time and space, we only illustrated the proposed external constraint handling method with direct collocation. We note that similar benefits should be obtainable, after some adaptations, with other simultaneous methods, such as direct multiple shooting. This might also be possible for sequential methods, such as direct single shooting. Moreover, it would be interesting to test a slightly altered version of our proposed constraint removal method with active-set based NLP solvers --- this could be compared under a mesh refinement scheme against earlier work \cite{chung2009external}, which exploits constraint addition instead.  



\bibliography{JournalBib2018} 

\begin{thebibliography}{10}

\bibitem{biegler2007overview}
L.~T. Biegler, ``An overview of simultaneous strategies for dynamic
  optimization,'' {\em Chemical Engineering and Processing: Process
  Intensification}, vol.~46, no.~11, pp.~1043--1053, 2007.

\bibitem{chung2009external}
H.~Chung, E.~Polak, and S.~Sastry, ``An external active-set strategy for
  solving optimal control problems,'' {\em IEEE Transactions on Automatic
  Control}, vol.~54, no.~5, pp.~1129--1133, 2009.

\bibitem{gondzio2016crash}
J.~Gondzio, ``Crash start of interior point methods,'' {\em European Journal of
  Operational Research}, vol.~255, no.~1, pp.~308--314, 2016.

\bibitem{jost2016constraint}
M.~Jost, G.~Pannocchia, and M.~M{\"o}nnigmann, ``Constraint removal in linear
  mpc: An improved criterion and complexity analysis,'' in {\em European
  Control Conference (ECC)}, pp.~752--757, IEEE, 2016.

\bibitem{jost2015accelerating}
M.~Jost, G.~Pannocchia, and M.~M{\"o}nnigmann, ``Accelerating tube-based model
  predictive control by constraint removal.,'' in {\em 54th IEEE Conference on
  Decision and Control (CDC)}, pp.~3651--3656, 2015.

\bibitem{constraintRemovalNMPC}
R.~Dyrska, , and M.~M{\"o}nnigmann, ``Accelerating nonlinear model predictive
  control by constraint removal,'' in {\em Proc.\ 6th IFAC Conference on
  Nonlinear Model Predictive Control}, 2018.

\bibitem{betts2010practical}
J.~T. Betts, {\em Practical Methods for Optimal Control and Estimation Using
  Nonlinear Programming: Second Edition}.
\newblock Advances in Design and Control, Society for Industrial and Applied
  Mathematics, 2010.

\bibitem{fahroo2008advances}
F.~Fahroo and I.~M. Ross, ``Advances in pseudospectral methods for optimal
  control,'' in {\em AIAA guidance, navigation and control conference and
  exhibit}, p.~7309, 2008.

\bibitem{liu2014hp}
F.~Liu, W.~W. Hager, and A.~V. Rao, ``An hp mesh refinement method for optimal
  control using discontinuity detection and mesh size reduction,'' in {\em 53th
  IEEE Conference on Decision and Control (CDC)}, pp.~5868--5873, IEEE, 2014.

\bibitem{fiacco1963programming}
A.~V. Fiacco and G.~P. McCormick, ``Programming under nonlinear constraints by
  unconstrained minimization: a primal-dual method,'' tech. rep., Research
  Analysis Corp Mclean VA, 1963.

\bibitem{fiacco1964sequential}
A.~V. Fiacco and G.~P. McCormick, ``The sequential unconstrained minimization
  technique for nonlinear programing, a primal-dual method,'' {\em Management
  Science}, vol.~10, no.~2, pp.~360--366, 1964.

\bibitem{fiacco1966extensions}
A.~V. Fiacco and G.~P. McCormick, ``Extensions of {SUMT} for nonlinear
  programming: equality constraints and extrapolation,'' {\em Management
  Science}, vol.~12, no.~11, pp.~816--828, 1966.

\bibitem{mehrotra1992implementation}
S.~Mehrotra, ``On the implementation of a primal-dual interior point method,''
  {\em SIAM Journal on optimization}, vol.~2, no.~4, pp.~575--601, 1992.

\bibitem{F50Performance}
Delft University of Technology, {\em Performance Model Fokker 50}, 2010.

\bibitem{ICLOCS2}
Y.~Nie, O.~J. Faqir, and E.~C. Kerrigan, ``{ICLOCS}2: Solve your optimal
  control problems with less pain,'' in {\em Proc.\ 6th IFAC Conference on
  Nonlinear Model Predictive Control}, 2018.

\bibitem{wachter2006implementation}
A.~W{\"a}chter and L.~T. Biegler, ``On the implementation of an interior-point
  filter line-search algorithm for large-scale nonlinear programming,'' {\em
  Mathematical programming}, vol.~106, no.~1, pp.~25--57, 2006.

\end{thebibliography}
\bibliographystyle{ieeetr}

\end{document}